%% file: main.tex
\documentclass[sigconf,nonacm]{acmart}
\usepackage{manuscript}

\begin{document}
\input{frontmatter.tex}

\input{sections/01_introduction.tex}
\input{sections/02_related_work.tex}
\input{sections/03_data_measurement.tex}
\input{sections/04_methodology.tex}
\input{sections/04_growth_accounting.tex}
\input{sections/05_author_population.tex}
\input{sections/06_individual_productivity.tex}
\input{sections/09_discussion.tex}
\input{sections/10_conclusion.tex}

\bibliographystyle{ACM-Reference-Format}
\bibliography{references}

\appendix
\input{appendices/a_external_comparison.tex}

\input{appendices/b_author_demography.tex}

\input{appendices/d_author_identity.tex}
\input{appendices/c_output_distribution.tex}

\input{appendices/e_sources_sample_fields.tex}
\input{appendices/f_accounting_proofs.tex}

\end{document}

%% file: frontmatter.tex
\title{Behind the Paper Explosion: Author Growth and Widening Gaps in Computer Science}

\author{Jiayi Li}
\affiliation{%
  \institution{National University of Singapore}
  \city{Singapore}
  \country{Singapore}
}
\author{Bowen Niu}
\affiliation{%
  \institution{National University of Singapore}
  \city{Singapore}
  \country{Singapore}
}
\author{Bingsheng He}
\affiliation{%
  \institution{National University of Singapore}
  \city{Singapore}
  \country{Singapore}
}
\author{Zhaomin Wu}
\affiliation{%
  \institution{National University of Singapore}
  \city{Singapore}
  \country{Singapore}
}

\begin{abstract}
The rapid growth of computer science (CS) papers increases the work needed to review research and find relevant studies. The growing use of AI tools in scientific writing makes it increasingly important to identify the authors behind this growth. Previous studies have tracked publication growth and followed researchers' careers, documenting the increasing number of authors and the many researchers who publish only briefly. However, annual averages of papers per author reflect different sets of researchers and may hide changes among the same authors. This study addresses this gap by linking individual changes to overall CS paper growth using arXiv, DBLP, and conference submission records for 2010--2025. With each paper divided equally among its authors, the total paper count equals the number of authors multiplied by their average papers per author. This relationship provides a decomposition of the two contributions of paper growth.
This decomposition attributes nearly all growth in CS paper volume to an increase in the number of publishing authors. The analysis of author numbers further distinguishes expansion through new and returning arXiv authors from the replacement of authors no longer publishing there. Changes in who publishes also affect the average, motivating comparisons of the same authors over time within groups with similar past paper counts. This comparison identifies substantial replacement within the growing author group. New authors average fewer papers under this counting method, keeping the overall average nearly unchanged despite increases among authors who published in both years. Moreover, statistics show that these increases are uneven, with widening gaps even among authors who previously published similar numbers of papers. These findings provide insights into peer review, community participation, researcher evaluation, and equity within the CS community.
\end{abstract}


\ccsdesc[500]{Information systems~Digital libraries and archives}
\ccsdesc[300]{Applied computing~Sociology}
\ccsdesc[300]{General and reference~Empirical studies}

\keywords{scholarly Web communities, computer science, scientific production, arXiv, author demography, growth attribution}

\maketitle

%% file: sections/01_introduction.tex
\section{Introduction}
\label{sec:intro}

\begin{figure}[ht]
\centering
\includegraphics[width=\columnwidth]{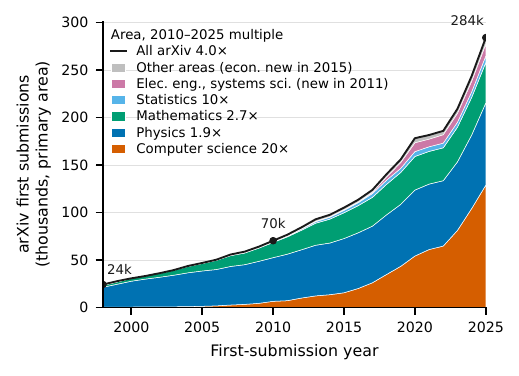}
\Description{Annual arXiv first submissions stacked by primary area from 1998 to 2025, showing the expansion of computer science within the archive.}
\caption{Annual arXiv first submissions by primary area.}
\label{fig:arxiv-area-composition}
\end{figure}

Scholarly Web communities face growing pressure to review and organize an expanding volume of computer science (CS) research \cite{Shah2018,Hanson2024,Kim2025}. Their platforms coordinate preprint sharing, conference peer review, and research discovery. CS has become the single largest driver of arXiv's expansion (Figure~\ref{fig:arxiv-area-composition}). Generative AI adds urgency to understanding this growth, as controlled experiments demonstrate substantial time savings in writing tasks and large-scale text analyses confirm widespread use of large language models (LLMs) in recent scientific manuscripts \cite{NoyZhang2023,Liang2025}. Effective responses depend on distinguishing expansion of the author community from changes in individual output. This distinction connects the growth in demand for peer review to the changing population that can contribute to it.

Prior studies have clarified how community expansion and collaboration contribute to publication growth. Bibliometric analyses show that the number of CS authors grows as fast as or faster than the number of papers, and that output per author under fractional or first-author counting stays nearly stable as teams grow \cite{Blackburn2019,Fanelli2016}. Career studies further establish that participation changes through substantial workforce turnover \cite{Milojevic2018,Boothby2022}, while recent work documents the rise of highly prolific authors \cite{AzadBanu2024}. Together, these findings establish the importance of both community membership and the distribution of author output. However, they leave unresolved how authors entering or leaving the publishing population affect average papers per author during the recent CS surge. A nearly stable mean can conceal changes among continuing authors, and annual rankings cannot establish whether authors with similar past output subsequently follow increasingly unequal trajectories.

To connect population growth with individual trajectories, we analyze CS papers from 2010 to 2025 using arXiv, DBLP publications, and ICLR submissions. An exact accounting identity expresses paper counts as the number of publishing authors multiplied by average paper output, separating author growth from changes in output per author under fractional and first-author counting. Comparisons across publication stages and with Mathematics and Physics assess whether the same pattern holds across sources and fields. We then trace entry, return, and inactivity to explain the net change in the arXiv-CS author population, using linked OpenAlex and ORCID records to distinguish platform entry from entry to research. Because turnover also changes the authors represented in the population mean, we define historical groups by 2018--2020 output and compare the same continuing authors through 2025 to assess whether authors with similar past output follow different trajectories.

Our analysis shows that growth in the number of publishing authors accounts for nearly all paper growth across preprints, conference submissions, and formal publications. Author growth remains the main source of paper growth on arXiv after 2022 and explains most of the faster growth of CS relative to Mathematics and Physics. Within arXiv CS, this expansion involves substantial replacement as new entrants and returning authors replace the output of authors inactive in 2025 and enlarge the population. Most entrants already have a publication history, so platform expansion draws heavily on researchers with prior experience. This turnover also keeps mean fractional output nearly constant because the large influx of lower-output entrants offsets gains among continuing authors. Those gains are uneven even among authors with similar past output, whose output differences widen.

The contribution of this paper can be summarized as below.
\begin{itemize}
\item We develop an analytical framework linking paper growth to author turnover and individual output trajectories.
\item We find that author expansion drives paper growth and output gaps widen even among authors with similar past output. We also report findings on author turnover, entry histories, and retention.
\item We derive insights from these findings for reviewer recruitment, researcher retention, and research evaluation in the CS community.
\end{itemize}

%% file: sections/02_related_work.tex
\section{Related Work}
\label{sec:related}

\input{tables/related_work_comparison.tex}

Studies of publication growth have shown that the number of authors and the way credit is shared explain much of the rise in paper counts. Long-run analyses document growing publication volume \cite{BornmannMutz2015} and larger research teams \cite{Wuchty2007}, and studies of CS find that author counts grow at least as fast as publication counts \cite{Cavero2014,Fernandes2017,Blackburn2019,Demetrescu2022}. Blackburn et al.\ also track fractional output per author and its percentiles \cite{Blackburn2019}, and Fanelli and Larivi\`ere find that the fractional and first-author output of early-career researchers has stayed nearly stable while their publication lists have grown longer \cite{Fanelli2016}. These studies, however, analyze accepted or formal publications, and their data end before 2022. They cannot show whether the same pattern holds for preprints and conference submissions or whether it continues through the recent surge. Table~\ref{tab:related-work-comparison} summarizes the scope of representative studies.

Research on scientific careers has documented high turnover in the workforce and many short publishing careers. Longitudinal analyses separate a continuously publishing core from transient researchers \cite{Ioannidis2014}, document shorter careers and a growing temporary workforce \cite{Milojevic2018}, and measure the entry, departure, and replacement of early-career researchers \cite{Boothby2022}. Work on CS and Web communities traces shifting topics, collaboration networks, faculty hiring, and movement across research domains \cite{Menezes2009,Ying2015,Way2016,Pramanik2019,Jin2021,GrauxOrlandi2022}. The career studies, however, observe entry and exit within a single bibliographic database, where an author's first appearance cannot show whether the author is new to research or has published elsewhere. Neither line of work relates these changes in membership to the growth in paper counts.

Research on prolific authors and AI assistance has identified changes in the upper tail of output and new ways to accelerate writing. Bibliometric studies document a surge of highly prolific authors \cite{Moreira2023,AzadBanu2024} and examine associations between AI adoption, submissions, and output \cite{Kusumegi2025,Tang2025,Matsui2026,Gartenberg2026}. However, annual counts of prolific authors compare changing populations and cannot establish whether the same authors publish more. Studies of AI assistance examine a possible source of rising output, with evidence of productivity gains in experiments and workplaces \cite{NoyZhang2023,Brynjolfsson2025}, while lexical analyses establish the growing presence of LLM-associated text in manuscripts and reviews \cite{Liang2024,Liang2025,Kobak2025,Qi2025,Cunningham2025,HeBu2026}. Author-level studies of AI adoption track changes among the same authors, although dating adoption by the first detected AI-assisted paper can select periods of high output and bias estimated gains upward \cite{Renault2026}. Connecting individual trajectories to aggregate growth further requires distinguishing changes among comparable authors from changes in population membership.

%% file: tables/related_work_comparison.tex
\begin{table*}
\centering
\caption{Analytical scope of representative studies. \ding{51} marks a covered dimension and \ding{55} a dimension outside the study's main analysis.}
\label{tab:related-work-comparison}
\begin{tabular*}{\textwidth}{@{\extracolsep{\fill}}lcccccc@{}}
\toprule
\textbf{\boldmath Study} & \textbf{\boldmath \shortstack{Paper and\\author counts}} & \textbf{\boldmath \shortstack{Fractional\\output}} & \textbf{\boldmath \shortstack{Entry\\cohorts}} & \textbf{\boldmath \shortstack{Prior publication\\histories}} & \textbf{\boldmath \shortstack{Historical\\output groups}} & \textbf{\boldmath \shortstack{Post-2022\\period}} \\
\midrule
Blackburn et al.\ \cite{Blackburn2019} & \ding{51} & \ding{51} & \ding{55} & \ding{55} & \ding{55} & \ding{55} \\
Fanelli and Larivi\`ere \cite{Fanelli2016} & \ding{55} & \ding{51} & \ding{51} & \ding{55} & \ding{55} & \ding{55} \\
Milojevi\'c et al.\ \cite{Milojevic2018} & \ding{55} & \ding{55} & \ding{51} & \ding{55} & \ding{55} & \ding{55} \\
Azad and Banu \cite{AzadBanu2024} & \ding{51} & \ding{55} & \ding{55} & \ding{55} & \ding{55} & \ding{55} \\
Kusumegi et al.\ \cite{Kusumegi2025} & \ding{55} & \ding{55} & \ding{55} & \ding{55} & \ding{55} & \ding{51} \\
\midrule
\textbf{Our study} & \ding{51} & \ding{51} & \ding{51} & \ding{51} & \ding{51} & \ding{51} \\
\bottomrule
\end{tabular*}
\end{table*}

%% file: sections/03_data_measurement.tex
\section{Problem Definition}
\label{sec:data}
\label{sec:problem-definition}
\label{sec:notation}
\label{sec:source-records}

We study how changes in author participation and individual output account for the rapid growth of CS papers from 2010 to 2025 (Figure~\ref{fig:arxiv-area-composition}). The central problem is to distinguish growth in the number of publishing authors from changes in output per author, and to determine how author turnover affects the relationship between this average and individual output. The rest of this section defines the scope of this problem and the terms it uses.

We study the CS papers recorded in three sources, which cover three stages of scholarly communication, namely preprints, conference submissions, and formal publications. On arXiv, our primary source, a CS paper is a preprint with at least one of the 40 \texttt{cs.*} categories, counted once across its versions and dated by the year of its first version. In DBLP, a CS paper is any journal article or conference paper that DBLP indexes, dated by its publication year and excluding CoRR preprints and informal records. For ICLR, a paper is a submission to one edition of the conference, whether or not it was accepted. Because cross-listed arXiv papers carry several categories, a paper can belong to more than one field. \emph{Broad AI} contains the CS papers with at least one of eight AI-related categories, and Mathematics and Physics are comparison fields defined in the same way (Appendix~\ref{sec:app-sources}). No source is filtered by country or institution. OpenAlex and ORCID supply authors' publications outside these sources and add no papers to the counts.

The papers of a source and field in year $t$ form the set $\mathcal{P}_t$, with \emph{paper count} $P_t=|\mathcal{P}_t|$, and \emph{growth} is the change in this count from an early period to a late period. We summarize each period by the geometric mean of its annual paper counts and report growth as the ratio of the late-period mean to the early-period mean. Section~\ref{sec:growth-decomposition} gives the corresponding logarithmic measure. Long-run comparisons use the periods 2010--2012 and 2023--2025 for arXiv and the 2017 and 2025 editions for ICLR. The DBLP late period is 2021--2023 because recent publications are still being indexed. A comparison of 2022 with 2025 covers the period after the release of ChatGPT, and analyses of individual authors compare 2021 with 2025. Figure~\ref{fig:arxiv-area-composition} differs from these measures because it assigns each arXiv record to its primary area and compares the single years 2010 and 2025.

We identify distinct authors across papers through identity resolution (Section~\ref{sec:identity}). The \emph{active authors} of a source, field, and year are the distinct authors of at least one of its papers. They form the set $\mathcal{A}_t$, with $A_t=|\mathcal{A}_t|$. An author's \emph{individual output} is the number of papers credited to the author in a year, and Section~\ref{sec:growth-decomposition} specifies how coauthored papers are credited. \emph{Output per author} is the average individual output of the active authors. An \emph{arXiv-CS entrant} is an author in the year of their first arXiv-CS paper, which we call their \emph{entry year}. Because an entrant may have published elsewhere before, their publications dated before the entry year in external records form their \emph{prior publication history}. An author is a \emph{first-time CS author} in the year of their first CS publication in arXiv or these records, and \emph{retention} describes whether first-time CS authors keep publishing CS papers in the following years.

%% file: sections/04_methodology.tex
\section{Methodology}
\label{sec:methodology}
\label{sec:framework}

An exact accounting identity separates paper growth into author growth and changes in output per author, which we interpret through population turnover and individual trajectories.

\subsection{Growth Decomposition}
\label{sec:analytical-methods}
\label{sec:growth-decomposition}

An exact decomposition requires each paper to contribute one unit of total credit. Let $\mathcal{B}(p)$ be the authors of paper $p$ and $n_p=|\mathcal{B}(p)|$ its \emph{team size}. \emph{Fractional counting} assigns each author an equal share $1/n_p$. Author $a$'s annual \emph{fractional output} is $q^A_{a,t}=\sum_{p\in\mathcal{P}_t:\,a\in\mathcal{B}(p)}1/n_p$, and \emph{mean fractional output} $Q^A_t$ averages this quantity over $\mathcal{A}_t$.

\emph{First-author counting} assigns each paper to the author listed first. The \emph{active first authors} form the set $\mathcal{F}_t$, with $F_t=|\mathcal{F}_t|$. An author's \emph{first-author output} $q^F_{a,t}$ counts papers listing that author first, and $Q^F_t$ averages this output over $\mathcal{F}_t$. These two averages are the \emph{population means} under their respective rules. Figures label fractional counting as all-author and first-author counting as first-listed.

The relation to team size follows from \emph{whole counting}, which credits each coauthor with one paper. The total number of author appearances is $H_t=\sum_{p\in\mathcal{P}_t}n_p$. Mean whole-count output is $W_t=H_t/A_t$, and mean team size is $\bar n_t=H_t/P_t$.

\begin{proposition}[Paper-output accounting identities]
\label{prop:paper-accounting}
For a finite, nonempty annual paper set with at least one recorded author per paper, the counting rules above satisfy
\begin{equation}
P_t=A_tQ^A_t=F_tQ^F_t,\qquad Q^A_t=\frac{W_t}{\bar n_t}.
\label{eq:central-decomposition}
\end{equation}
\end{proposition}

These annual identities provide the basis for comparing changes across periods. For a positive annual quantity $X_t$, define $g_X=\log\bar X_{\tau_1}-\log\bar X_{\tau_0}$, where $\bar X_\tau$ is the geometric mean over the years of period $\tau$. The growth multiple is $e^{g_X}$. Annualized growth divides $100\,g_X$ log points by the years between period midpoints.

\begin{corollary}[Growth decomposition]
\label{cor:growth-decomposition}
For positive annual quantities evaluated over the same comparison periods, paper growth satisfies
\begin{align}
g_P&=g_A+g_{Q^A},\label{eq:all-author}\\
g_P&=g_F+g_{Q^F}.\label{eq:first-author}
\end{align}
For $g_P\neq0$, the \emph{contribution shares} satisfy
\begin{equation}
\begin{aligned}
\frac{g_A}{g_P}+\frac{g_{Q^A}}{g_P}&=1,\\
\frac{g_F}{g_P}+\frac{g_{Q^F}}{g_P}&=1.
\end{aligned}
\label{eq:contribution-shares}
\end{equation}
The first term represents author growth and the second the change in mean output. A share may be negative or exceed one when the other component offsets growth.
\end{corollary}

Proofs are given in Appendix~\ref{sec:app-accounting-proofs}. The team-size identity also explains why mean fractional output falls when team size grows faster than the number of papers each author appears on.

We compute the contribution shares for each source and arXiv field. ICLR authors are identified by their names and are analyzed under fractional counting only. Because Eq.~\eqref{eq:all-author} holds in every field, the gap in annualized paper growth between two fields equals the gap in author growth plus the gap in the change in output per author, and we report the share of each part. To locate growth within CS, we divide each paper equally among its listed CS categories and compute each category's share of the increase in CS papers.

Author growth measures net expansion but does not distinguish new participation from replacement. The turnover analysis assigns each author active in arXiv CS in 2021 or 2025 to one \emph{activity path}. \emph{Continuing authors} are active in both years. \emph{New entrants} have an entry year after 2021 and are active in 2025. \emph{Returning authors} entered before 2021, have no arXiv-CS paper in 2021, and are active in 2025. \emph{Authors inactive in 2025}, or \emph{inactive authors}, are active in 2021 and have no arXiv-CS paper in 2025. Inactivity refers only to arXiv CS in 2025. New entrants and returning authors form the \emph{inflow}, and the four paths describe the turnover of the author population. Summing the paths' changes in fractional output recovers the change in papers from 2021 to 2025. To place the paths in the output distribution, we define two \emph{output tiers}. The \emph{high-output tier} contains authors with individual output above four under the counting rule analyzed ($q^A>4$ or $q^F>4$), and the \emph{low-output tier} contains authors with fractional output below one quarter ($q^A<0.25$).

Platform entry need not mark the start of a research career, so we classify an entrant's prior publication history as prior CS if it contains any CS publication, prior non-CS if all its publications lie outside CS, and no prior record observed if it is empty. Entrants without a linked external profile remain unresolved. To assess whether entry leads to sustained participation, we measure retention among first-time CS authors, grouped into \emph{entry cohorts} by the year of their first CS publication. We use the 2010--2022 cohorts, which have three observable follow-up years before 2025. For an author whose first CS publication appears in year $s$, let $y_k=1$ if the author has a CS paper in year $s+k$ and $y_k=0$ otherwise. \emph{Return} is $\max_{1\le k\le3}y_k$, \emph{Year~3} is $y_3$, and \emph{Active years} is $\frac13\sum_{k=1}^{3}y_k$. Each reported rate averages an outcome within each cohort and gives the thirteen cohorts equal weight. We compare authors whose entry paper includes an \emph{established CS coauthor}, a coauthor with an earlier CS publication, with authors on \emph{newcomer-only teams}, whose authors are all first-time CS authors.

The population mean also reflects changes in who publishes. Comparisons within \emph{historical output groups}, defined before the comparison period, examine how authors with similar past output subsequently differ. Under each counting rule, we rank the authors with positive output in at least one year of 2018--2020 by their total output over these years and divide the ranking into twenty five-percentile groups, with five one-percentile groups within the historical Top~5\%. In each later year, the \emph{active members} of a group are its members with positive output, and the \emph{group mean} averages their output. Because active membership can change, we also compare the \emph{same continuing authors}, the group members with positive output in both 2021 and 2025. To distinguish widening differences within historical groups from divergence between them, we split output variance among active members into mean within-group variance and the variance of group means, weighting both by the groups' shares of active members. Among the same continuing authors, the shares whose output rises or falls show how widespread the changes are. The mean, median, and ninetieth percentile describe output levels, while the Gini coefficient measures inequality. We also calculate the mean, variance, and Gini coefficient for the full annual population to assess whether the patterns within historical groups extend to all active authors. To relate the output distribution to the population mean, we express each output tier's contribution as its total output divided by the number of all active authors.

\subsection{Data Collection, Cleaning, and Validation}
\label{sec:identity}

We build annual records of papers and author output from the metadata of each source. For arXiv, we use the Cornell metadata snapshot version~295 (July~18, 2026), which contains 3.11 million records, and keep records with an identifier, a first-version (\texttt{v1}) date, and a usable author list. Reports and reviews posted to arXiv count as papers, and author lists and categories take their values in the snapshot. For DBLP, we use publication years and exclude CoRR preprints and informal records. For ICLR, we use archived OpenReview-derived submission records, assign each submission to its conference edition, whose deadline falls in the preceding year, and calibrate the counts to official totals. External publication records come from the June~2026 OpenAlex snapshot and its updates, ORCID identifiers, and DBLP links.

Identity resolution links the records of the same author across papers and sources. We remove duplicate arXiv records and duplicate authors within a paper while preserving author order, and we standardize names for character encoding, case, accents, and punctuation. DBLP merges two name variants only when at least two works linked by DBLP records or digital object identifiers (DOIs) support the same profile. For arXiv-CS entrants, we group candidate OpenAlex profiles that belong to the same researcher and accept a match through shared identifiers or through a calibrated trajectory model that compares publication histories (Appendix~\ref{sec:app-identity-resolution}). Prior publication histories use the deduplicated works of accepted profiles, and retention uses only works attributable to the individual author. Entrants without an accepted match remain in the arXiv counts with unresolved histories.

We validate the resulting records in four ways. We check paper totals, author appearances, and the accounting identities, calibrate ICLR against official counts, evaluate identity linkage on held-out authors excluded from model training, and measure how name resolution changes author counts (Appendices~\ref{sec:app-external}, \ref{sec:app-identity-resolution}, and~\ref{sec:app-sources}).

%% file: sections/04_growth_accounting.tex
\section{Analysis and Results}
\label{sec:analysis}
\label{sec:internal}

The aggregate decomposition identifies which component accounts for paper growth but leaves the author-level changes behind it to be explained. Section~\ref{sec:growth-accounting} establishes the contributions of author growth and output per author. To explain the dominant contribution of author growth, Section~\ref{sec:population} examines the turnover behind net expansion. Because turnover also changes the population represented by average output, Section~\ref{sec:productivity} relates this average to changes among continuing authors and examines their trajectories within groups defined by prior output.

\subsection{Active Authors and Individual Output}
\label{sec:growth-accounting}
\label{sec:author-growth}
\label{sec:byline-robustness}
\label{sec:comparison}

\begin{figure*}
\centering
\includegraphics[width=\textwidth]{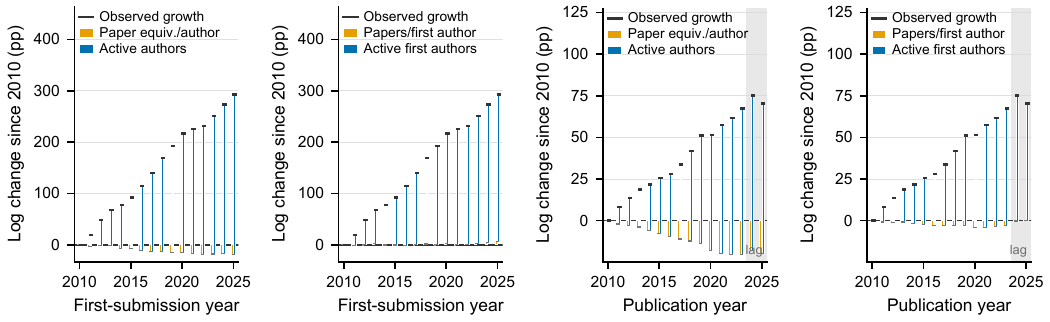}
\begin{tabularx}{\textwidth}{@{}*{4}{>{\centering\arraybackslash}X}@{}}
(a) arXiv CS, all authors & (b) arXiv CS, first-listed & (c) DBLP CS, all authors & (d) DBLP CS, first-listed
\end{tabularx}
\Description{Four panels decompose cumulative CS paper growth relative to 2010 into active-author growth and changes in output per author. The first two panels show arXiv CS under fractional and first-author counting, and the last two show DBLP under the same rules. Black marks show observed paper growth, blue segments show author growth, and orange segments show changes in output per author. Gray shading covers 2024--2025 in DBLP because these years are incompletely indexed.}
\caption{Cumulative contributions of author growth and output per author to CS paper growth. Shading marks recent DBLP years with incomplete indexing.}
\label{fig:cs-accounting}
\end{figure*}

\textbf{Growth in active authors accounts for nearly all long-run growth in CS papers, across preprints, conference submissions, and formal publications.} In arXiv CS and DBLP, author growth dominates cumulative paper growth under both counting rules throughout the study period (Figure~\ref{fig:cs-accounting}). Broad AI shows the same pattern, with author shares of $102.7\%$ under fractional counting and $97.4\%$ under first-author counting. For ICLR submissions, the fractional-counting share is $101.4\%$ (Appendix~\ref{sec:app-author-growth}).

\input{sections/productivity_events.tex}

\textbf{Faster author growth also accounts for most of the advantage of CS and Broad AI over Mathematics and Physics.} Differences in author growth account for $94\%\text{--}98\%$ of the gaps in paper growth between CS or Broad AI and the two comparison fields, and for $95\%$ of Broad AI's advantage over CS as a whole (Table~\ref{tab:growth-gaps}). Mathematics and Physics also show growing teams and declining fractional output per author (Appendix~\ref{sec:app-accounting}).

\input{tables/cross_field_growth_gaps.tex}

Within CS, \textbf{the faster-growing AI-related subfields produce over 70\% of the increase in papers.} When each paper is divided equally among its listed CS categories, the eight Broad AI categories supply $70.5\%$ of the growth in arXiv-CS papers from 2010--2012 to 2023--2025, and machine learning and computer vision together supply nearly two-fifths (Figure~\ref{fig:subfields}).

\begin{figure}
\centering
\includegraphics[width=\columnwidth]{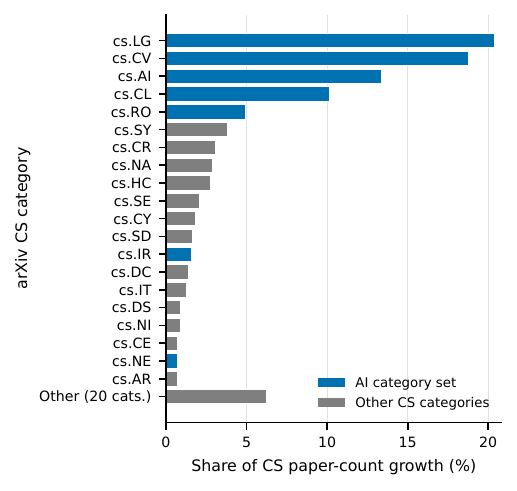}
\Description{CS category contributions to paper-count growth using fractional allocation across each paper's listed CS categories.}
\caption{CS subfield contributions to paper-count growth. Each paper is divided equally across its listed CS categories.}
\label{fig:subfields}
\end{figure}

%% file: sections/productivity_events.tex
\label{sec:events}

\textbf{The dominance of author growth continues after the release of ChatGPT.} From 2022 to 2025, a period of rapid uptake of LLM-assisted writing in science \cite{Liang2025}, the number of authors continues to supply most of the new papers under both counting rules (Table~\ref{tab:recent-growth-accounting}). Mean first-author output contributes a larger share than in the long-run comparison, so the recent surge combines continued author growth with a modest rise in output per author, which Section~\ref{sec:productivity} traces to individual authors.

\input{tables/recent_growth_accounting.tex}

%% file: tables/recent_growth_accounting.tex
\begin{table}
\centering
\caption{Contributions to CS log paper growth, 2022--2025.}
\label{tab:recent-growth-accounting}
\begin{tabular*}{\columnwidth}{@{\extracolsep{\fill}}lrr@{}}
\toprule
\textbf{\boldmath Convention} & \textbf{\boldmath Author share (\%)} & \textbf{\boldmath Mean share (\%)} \\
\midrule
Fractional & 98.9 & 1.1 \\
First-author & 91.1 & 8.9 \\
\bottomrule
\end{tabular*}
\end{table}

%% file: tables/cross_field_growth_gaps.tex
\begin{table}
\centering
\caption{Decomposition of growth differences across fields under fractional counting, 2010--2012 to 2023--2025.}
\label{tab:growth-gaps}
\begin{tabularx}{\columnwidth}{@{}>{\raggedright\arraybackslash}Xrrrrr@{}}
\toprule
\textbf{\boldmath Field pair} & \textbf{\boldmath Gap} & \multicolumn{2}{c}{\textbf{\boldmath Log pts/yr}} & \multicolumn{2}{c}{\textbf{\boldmath Gap share (\%)}} \\
\cmidrule(lr){3-4} \cmidrule(l){5-6}
 & \textbf{\boldmath (pts)} & \textbf{\boldmath Authors} & \textbf{\boldmath \shortstack{Per\\auth.}} & \textbf{\boldmath Authors} & \textbf{\boldmath \shortstack{Per\\auth.}} \\
\midrule
CS$-$Math & 13.36 & 13.14 & 0.22 & 98.35 & 1.65 \\
CS$-$Physics & 15.38 & 14.48 & 0.90 & 94.15 & 5.85 \\
Broad AI$-$CS & 9.53 & 9.06 & 0.48 & 94.98 & 5.02 \\
Broad AI$-$Math & 22.89 & 22.19 & 0.70 & 96.95 & 3.05 \\
Broad AI$-$Physics & 24.92 & 23.54 & 1.38 & 94.47 & 5.53 \\
\bottomrule
\end{tabularx}
\end{table}

%% file: sections/05_author_population.tex
\subsection{Growth in Active Authors}
\label{sec:population}
\label{sec:groups}
\label{sec:author-paths}
\label{sec:origins}
\label{sec:provenance}
\label{sec:recovered-histories}
\label{sec:retention}

\textbf{Net author growth conceals rapid turnover in the arXiv-CS author population.} The number of active authors nearly doubles from 2021 to 2025, while more than half of the 2021 authors publish no arXiv-CS paper in 2025 (Table~\ref{tab:author-path-output}). The inflow of new entrants and returning authors is therefore $57\%$ larger than the net increase in authors, so it replaces inactive authors and enlarges the population at the same time. By 2025, new entrants make up about two-thirds of all active authors, and continuing and returning authors make up the remaining third.

\input{tables/author_path_output.tex}

\begin{figure}
\centering
\includegraphics[width=\columnwidth]{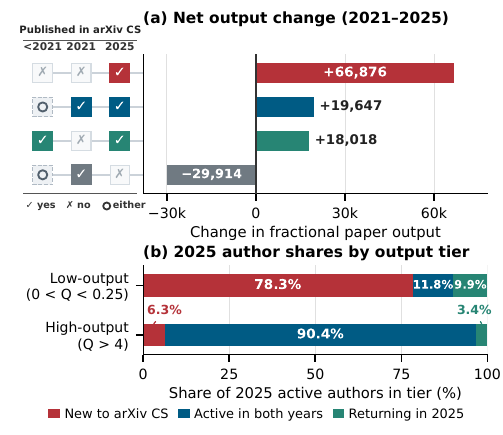}
\Description{Signed changes in fractional output by four author paths and stacked author shares in the low- and high-output tiers. New and returning authors supply net additions, while continuing authors dominate the high-output tier.}
\caption{Activity-path contributions to the 2021--2025 change in fractional output and path shares of the 2025 output tiers.}
\label{fig:bowen-trajectories}
\end{figure}

Because the output that inactive authors produced in 2021 is lost, \textbf{new entrants and returning authors supply more than the entire net increase in fractional output.} Together they contribute $113.8\%$ of the net gain from 2021 to 2025, which more than offsets the loss from inactive authors (Figure~\ref{fig:bowen-trajectories}a). Continuing authors add output as well, but most of the added output comes from the inflow.

\textbf{Most members of the low-output tier are new entrants, while most members of the high-output tier were already active in 2021.} The inflow contributes through its size despite its lower average output (Table~\ref{tab:author-path-output}). New entrants make up over three-quarters of the low-output tier, whereas about nine in ten members of the high-output tier ($q^A>4$) in 2025 were active in 2021 (Figure~\ref{fig:bowen-trajectories}b). Nearly three-fifths of the high-output tier's 2025 output comes from continuing authors who produced at most four fractional papers in 2021 (Table~\ref{tab:high-tail-transitions}).

\input{tables/high_tail_transitions.tex}

\begin{figure}
\centering
\includegraphics[width=\columnwidth]{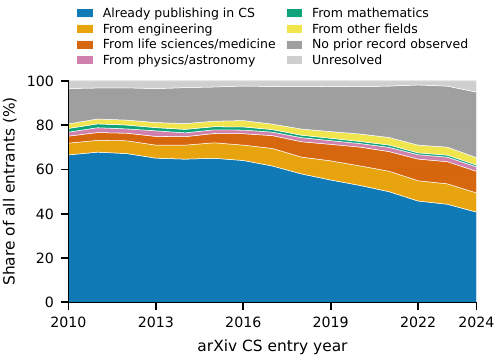}
\Description{Annual shares of arXiv-CS entrants by prior publication history, including prior CS, other fields, no observed prior publication, and unresolved history.}
\caption{Prior publication histories of arXiv-CS entrants. Prior works precede the entry calendar year.}
\label{fig:entrant-composition}
\end{figure}

Although new to arXiv CS, \textbf{roughly three-quarters of all arXiv-CS entrants from 2010 to 2024 already have a prior publication history.} Entrants who already publish in CS form the largest group (Figure~\ref{fig:entrant-composition}). Among entrants whose prior publications all lie outside CS, life sciences and medicine together with engineering account for nearly three-quarters. Entrants with no prior record observed and entrants with prior non-CS publications grow fastest, about three times as fast as entrants with prior CS publications between 2010 and 2024 (Table~\ref{tab:entrant-history-summary}). Among the 2024 entrants without earlier arXiv papers that can be linked to DBLP, over two-fifths already have formal CS publications (Appendix~\ref{sec:app-prior-record}).

Among authors publishing in CS for the first time, \textbf{collaboration with an established CS coauthor at entry is associated with more sustained participation.} In the complete 2010--2022 entry cohorts, authors whose entry paper includes an established CS coauthor are 19.2 percentage points more likely to return within three years than authors on newcomer-only teams, a ratio of $1.37$ (Table~\ref{tab:bowen-entry-return}). They are also about one and a half times as likely to publish in the third year and are active in $1.49$ times as many follow-up years.

%% file: tables/author_path_output.tex
\begin{table}
\centering
\caption{Author counts and mean fractional output by activity path, 2021--2025. Dashes mark an inactive endpoint.}
\label{tab:author-path-output}
\begin{tabular*}{\columnwidth}{@{\extracolsep{\fill}}lrrr@{}}
\toprule
\textbf{\boldmath Path} & \textbf{\boldmath Authors} & \textbf{\boldmath Mean 2021} & \textbf{\boldmath Mean 2025} \\
\midrule
New entrants & 203,928 & -- & 0.328 \\
Returning & 36,574 & -- & 0.493 \\
Continuing & 68,081 & 0.699 & 0.988 \\
Inactive in 2025 & 87,539 & 0.342 & -- \\
\bottomrule
\end{tabular*}
\end{table}

%% file: tables/high_tail_transitions.tex
\begin{table}
\centering
\caption{Members of the 2025 high-output tier ($q^A>4$) by their 2021 output. Shares use the tier's total fractional output in 2025.}
\label{tab:high-tail-transitions}
\begin{tabular*}{\columnwidth}{@{\extracolsep{\fill}}lrrr@{}}
\toprule
\textbf{\boldmath 2021 position} & \textbf{\boldmath Authors} & \textbf{\boldmath \shortstack{2025\\output}} & \textbf{\boldmath Share (\%)} \\
\midrule
Active in 2021, $0<q^A\leq4$ & 1,712 & 10,154.7 & 58.5 \\
Active in 2021, $q^A>4$ & 558 & 5,893.3 & 34.0 \\
Not active in 2021 & 242 & 1,306.2 & 7.5 \\
\bottomrule
\end{tabular*}
\end{table}

%% file: sections/06_individual_productivity.tex
\subsection{Changes in Individual Output}
\label{sec:productivity}
\label{sec:team-productivity}
\label{sec:productivity-distribution}
\label{sec:participation-concentration}
\label{sec:career-composition}

\begin{figure*}
\centering
\includegraphics[width=\textwidth]{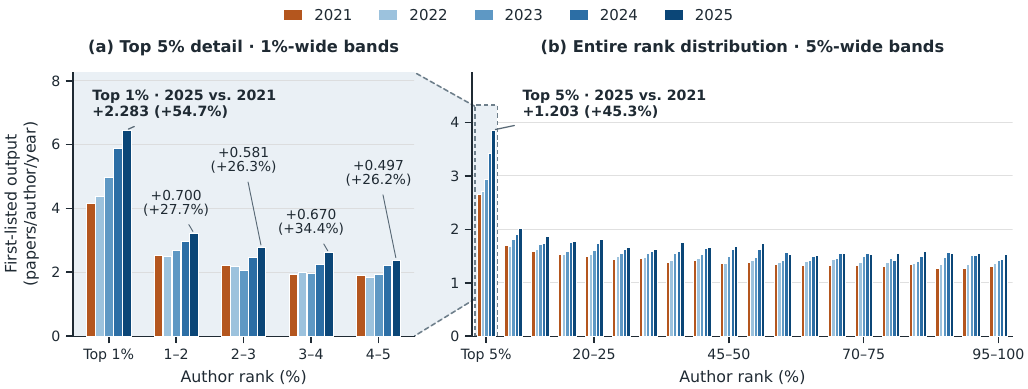}
\Description{First-author output among each year's active members of fixed historical output groups, including a detailed view of the historical top five percent.}
\caption{First-author output across historical output groups. Ranks are fixed by 2018--2020 output, and annual means include active group members.}
\label{fig:bowen-rank-bands}
\end{figure*}

\textbf{The lower output of new entrants keeps mean fractional output nearly unchanged while continuing authors produce more.} Continuing authors raise their mean fractional output by $41.3\%$ from 2021 to 2025, whereas new entrants in 2025 average about one-third of a continuing author's output and returning authors fall in between (Table~\ref{tab:author-path-output}). Because new entrants are the largest group of 2025 authors (Section~\ref{sec:population}), their lower output carries a large weight in the population mean.

Across groups of authors with similar past output, \textbf{active members of every historical group raise their mean first-author output from 2021 to 2025, with the largest gains at the top.} First-author counts measure these changes without dividing each paper's credit among coauthors. Active members of the historical Top~$5\%$ raise their group mean by nearly half, and those of the historical Top~$1\%$ raise it by more than half, roughly twice the relative gain of percentiles 1--5 (Figure~\ref{fig:bowen-rank-bands}).

The rising group means come with widening differences, and \textbf{most of the increase in output variance occurs within historical output groups, among authors with similar past output.} Within-group variance accounts for about nine-tenths of the increase in variance among active group members from 2021 to 2025 under both counting rules (Table~\ref{tab:historical-variance-growth}), and it still dominates when group weights are fixed at their 2021 values (Appendix~\ref{sec:app-distribution}). Because the active members of a group change from year to year, this comparison combines changes in individual output with changes in membership.

Holding membership fixed separates the two, and \textbf{the same continuing authors become more unequal within every historical output group.} The Gini coefficient of their output rises in all twenty groups under both counting rules. In the historical fractional Top~$5\%$, $56\%$ of these authors increase their output and $43\%$ reduce it, and first-author output also rises for some members and falls for others in every group (Appendix~\ref{sec:app-distribution}).

\begin{figure*}
\centering
\includegraphics[width=\textwidth]{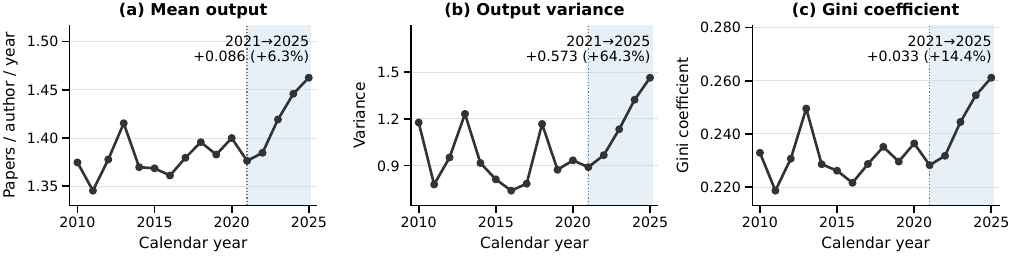}
\Description{Annual mean, population variance, and Gini coefficient of individual first-author output. Variance and Gini increase sharply from 2021 to 2025 while the mean rises more modestly.}
\caption{Annual mean and dispersion of first-author CS output. Each year includes authors with positive first-author output.}
\label{fig:bowen-annual-trends}
\end{figure*}

The uneven gains and losses are reflected in \textbf{a much faster rise in the mean than in the median within historical Top~5\% groups.} The same continuing authors in the historical first-author Top~$5\%$ raise their mean output by about a third from 2021 to 2025, while their median stays at two papers per year and their ninetieth percentile rises by two-thirds. In this group, authors whose output increases gain about twice as many papers on average as declining authors lose. The historical fractional Top~$5\%$ shows a similar contrast, with the mean rising by about a third and the median by less than one-tenth.

\input{tables/historical_variance_growth.tex}

The same divergence appears in the full annual population, \textbf{where the variance of first-author output grows ten times as fast as its mean from 2021 to 2025.} The Gini coefficient of first-author output also rises sharply (Figure~\ref{fig:bowen-annual-trends}), and under both counting rules, variance and the Gini coefficient record their largest four-year increases since 2010 (Appendix~\ref{sec:app-distribution}).

As output becomes more dispersed, \textbf{the high-output tier generates 70.6\% of the 2021--2025 increase in mean first-author output.} Under fractional counting, the high-output tier also raises the mean, but lower tiers offset this gain and leave mean fractional output nearly unchanged (Appendix~\ref{sec:app-mean-components}). The growing low-output share also extends beyond recent entrants, rising within every group of platform tenure, the years since an author's arXiv-CS entry (Appendix~\ref{sec:app-tenure}).

\input{tables/cs_output_margins.tex}

\textbf{Over the long run, team size grows faster than the number of papers each author appears on, so mean fractional output declines.} This longer-term decline contrasts with the nearly stable mean in 2021--2025. From 2010--2012 to 2023--2025, an active CS author's mean whole-count output rises by more than half, but mean team size grows 1.5 times as fast (Table~\ref{tab:cs-output-margins}). Because $Q^A=W/\bar n$ (Eq.~\eqref{eq:central-decomposition}), faster team growth lowers fractional output per author even as authors appear on more papers. DBLP and ICLR show the same combination of rising whole-count output, faster team growth, and declining fractional output (Appendix~\ref{sec:app-productivity}).

%% file: tables/historical_variance_growth.tex
\begin{table}
\centering
\caption{Variance growth among active members of historical output groups, 2021--2025.}
\label{tab:historical-variance-growth}
\begin{tabular*}{\columnwidth}{@{\extracolsep{\fill}}lrrr@{}}
\toprule
\textbf{\boldmath Credit} & \textbf{\boldmath \shortstack{Within\\groups}} & \textbf{\boldmath \shortstack{Between\\groups}} & \textbf{\boldmath \shortstack{Within\\share}} \\
\midrule
Fractional & 1.325 & 0.164 & 89.0\% \\
First-author & 3.303 & 0.371 & 89.9\% \\
\bottomrule
\end{tabular*}
\end{table}

%% file: tables/cs_output_margins.tex
\begin{table}
\centering
\caption{Whole-count output, team size, and mean fractional output in arXiv CS. Early and late are geometric means for 2010--2012 and 2023--2025.}
\label{tab:cs-output-margins}
\begin{tabular*}{\columnwidth}{@{\extracolsep{\fill}}lrrr@{}}
\toprule
\textbf{\boldmath Measure} & \textbf{\boldmath Early} & \textbf{\boldmath Late} & \textbf{\boldmath Change} \\
\midrule
Whole appearances/author $W$ & 1.54 & 2.36 & $+53.1\%$ \\
Authors/paper $\bar n$ & 2.64 & 4.75 & $+80.0\%$ \\
Fractional output/author $Q^A$ & 0.58 & 0.50 & $-14.9\%$ \\
\bottomrule
\end{tabular*}
\end{table}

%% file: sections/09_discussion.tex
\section{Discussion}
\label{sec:discussion}
\label{sec:projection}

Responding to publication growth requires attention to both community expansion and differences in individual trajectories. These two dimensions inform reviewing and participation policies, expectations in researcher evaluation, and efforts to ensure that the benefits of AI are shared equitably.

Author expansion makes the growing population of publishing researchers a potential source of additional reviewers. \textbf{This provides a rationale for reciprocal reviewing requirements as a potentially sustainable response to submission growth.} Such requirements have been adopted by ICML 2026 and ICLR 2027 \cite{ICML2026CallForPapers,ICLR2027AuthorGuidelines}. When growth comes primarily from more authors, the community gains potential reviewing capacity alongside additional review demand. Conferences can mobilize this capacity by recruiting qualified submitting authors, training new reviewers, and distributing assignments across the expanding author pool. The key is to convert author growth into completed, high-quality reviews through appropriate expertise matching and manageable workloads.

Given this potential reviewer supply, \textbf{entry restrictions should be assessed in terms of both potential workload relief and the risk of excluding capable contributors or constraining future reviewer supply.} ICLR 2027 permits each author to appear on at most one submission without an eligible reciprocal reviewer among its authors \cite{ICLR2027AuthorGuidelines,Andreas2026ICLRPolicies}. arXiv's January 2026 policy requires both an institutional email and prior authorship in the relevant arXiv domain for automatic endorsement, with personal endorsement as an alternative \cite{Boboris2026Endorsement}. These measures target reviewing capacity and moderation quality, respectively. Our finding that most arXiv-CS entrants already have publication histories highlights why their implementation should distinguish platform entry from research inexperience. This favors procedures that recognize relevant experience acquired outside the platform, while providing training and mentoring for genuine newcomers.

At the individual level, \textbf{rising average paper counts alone do not justify uniformly higher publication expectations in hiring and funding decisions.} Section~\ref{sec:productivity} shows widening output gaps even among continuing authors with similar past output. In the historical first-author Top~$5\%$, the group mean rises by about a third while the median author's output does not change. Applying this average increase as a common benchmark would therefore set expectations that do not reflect typical individual trajectories.

\textbf{Beyond accounting for these disparities in evaluation, we need to understand their causes, including the potential role of AI.} Potential explanations include unequal access to advanced models and computing resources, as well as differences in the skills needed to use AI effectively in research. Distinguishing affordability barriers from skills gaps could help identify appropriate responses. Future work could link researchers' AI access, usage, and training to their individual output trajectories and evaluate interventions such as subsidized access, shared infrastructure, and practical training. The goal is to enable researchers with different resources and skill levels to benefit equitably from AI.

%% file: sections/10_conclusion.tex
\section{Conclusion}
\label{sec:conclusion}

The computer science paper explosion shows how community expansion can sustain rapid paper growth while masking changes in individual output. Across arXiv preprints, ICLR submissions, and DBLP publications, growth in the number of publishing authors accounts for nearly all paper growth. Within arXiv CS, this expansion involves substantial turnover, with lower-output entrants keeping the population mean nearly unchanged despite a rise of over forty percent in continuing authors' mean fractional output. Even among continuing authors with similar past output, widening differences show that these gains are uneven. Understanding publication growth therefore requires following both the changing author population and the trajectories of its members, informing reviewer recruitment and evaluations that account for authors' prior output.

%% file: appendices/a_external_comparison.tex
\section{Growth Accounting Across Sources and Fields}
\label{sec:app-external}

This appendix compares the arXiv results of Sections~\ref{sec:growth-accounting} and~\ref{sec:productivity} with DBLP, ICLR, and the comparison fields. The DBLP data contain 5,084,864 formal publications from 2010 to 2025.

\subsection{Author Growth in DBLP and ICLR}
\label{sec:app-author-growth}
\label{sec:app-byline-robustness}

Under fractional counting, author growth contributes $134\%$ of paper growth in DBLP. Author growth also exceeds paper growth in the conference and journal subsets of DBLP, so mean fractional output contributes negatively in both late periods. First-author growth also dominates across these populations, as mean first-author output contributes negatively in the full corpus and conference subset while making a small positive contribution in the journal subset (Table~\ref{tab:dblp-arxiv-accounting}). In ICLR, submissions grow $23.9\times$ from 2017 to 2025, while the number of distinct submitting authors grows $25.0\times$, leaving mean fractional output nearly constant.

\input{tables/dblp_arxiv_accounting.tex}

\subsection{Preprints and Formal Publications}
\label{sec:growth-magnitudes}
\label{sec:venue-context}
\label{sec:app-source-coverage}
\label{sec:app-comparison-series}

The arXiv-CS paper count grows much faster than the DBLP count of formal CS publications from 2010 to 2023 ($12.3\times$ versus $2.0\times$). Over the same years, the share of active DBLP CS authors with a linked arXiv preprint rises by 25 percentage points, so wider use of preprints among formally publishing authors contributes to the faster growth of arXiv.

\subsection{Team Size and Output Across Fields and Sources}
\label{sec:app-accounting}
\label{sec:app-productivity}

In Broad AI, Mathematics, Physics, DBLP, and ICLR, authors appear on more papers per year, but mean team size grows faster, so mean fractional output declines in every case (Table~\ref{tab:whole-count-productivity}).

\input{tables/coauthorship_validation.tex}

%% file: tables/dblp_arxiv_accounting.tex
\begin{table}
\centering
\caption{Contribution shares in DBLP and arXiv over matched periods. The early period is 2010--2012, and the main late period is 2021--2023.}
\label{tab:dblp-arxiv-accounting}
\begin{tabular*}{\columnwidth}{@{\extracolsep{\fill}}lrrrr@{}}
\toprule
\textbf{\boldmath Quantity} & \textbf{\boldmath \shortstack{arXiv\\CS}} & \textbf{\boldmath \shortstack{DBLP\\all}} & \textbf{\boldmath \shortstack{DBLP\\conf.}} & \textbf{\boldmath \shortstack{DBLP\\jour.}} \\
\midrule
Growth (log pts/yr) & 19.40 & 4.99 & 2.79 & 7.60 \\
\quad Active authors (\%) & 107.81 & 133.96 & 174.65 & 113.49 \\
\quad Output per author (\%) & $-7.81$ & $-33.96$ & $-74.65$ & $-13.49$ \\
\midrule
\multicolumn{5}{@{}l}{\emph{First authors only (2010--2012 to 2021--2023)}} \\
\quad First authors (\%) & 99.07 & 105.14 & 117.16 & 96.56 \\
\quad Papers per first author (\%) & $+0.93$ & $-5.14$ & $-17.16$ & $+3.44$ \\
\midrule
\multicolumn{5}{@{}l}{\emph{Late period 2023--2025}} \\
Growth (log pts/yr) & 19.23 & 4.90 & 2.93 & 7.20 \\
\quad Active authors (\%) & 106.46 & 128.00 & 162.11 & 110.79 \\
\quad Output per author (\%) & $-6.46$ & $-28.00$ & $-62.11$ & $-10.79$ \\
\bottomrule
\end{tabular*}
\end{table}

%% file: tables/coauthorship_validation.tex
\begin{table}
\centering
\caption{Changes in whole-count output, fractional output, and team size across fields and sources.}
\label{tab:whole-count-productivity}
\begin{tabular*}{\columnwidth}{@{\extracolsep{\fill}}lrrr@{}}
\toprule
\textbf{\boldmath Population} & \textbf{\boldmath $W$} & \textbf{\boldmath $Q^A$} & \textbf{\boldmath $\bar n$} \\
\midrule
\multicolumn{4}{@{}l}{\emph{arXiv, 2010--2012 to 2023--2025}} \\
Broad AI & $+69.5\%$ & $-9.5\%$ & $+87.2\%$ \\
Mathematics & $+3.7\%$ & $-17.3\%$ & $+25.4\%$ \\
Physics & $+17.7\%$ & $-24.3\%$ & $+55.4\%$ \\
\midrule\multicolumn{4}{@{}l}{\emph{External sources, endpoint years}} \\
DBLP, 2010--2023 & $+9.9\%$ & $-18.5\%$ & $+34.8\%$ \\
ICLR, 2017--2025 & $+37.3\%$ & $-4.5\%$ & $+43.7\%$ \\
\bottomrule
\end{tabular*}
\end{table}

%% file: appendices/b_author_demography.tex
\section{Entrant Histories and Retention}
\label{sec:app-demography}

This appendix supports Section~\ref{sec:population}. The analysis of prior publication histories follows \emph{author keys}, which identify arXiv-CS entrants in the arXiv records, whereas the retention analysis follows deduplicated first-time CS authors.

Entrants with prior non-CS publications are classified by the fields of those publications. Among the 26,104 such entrants in 2024, life sciences and medicine account for $39.6\%$, engineering for $35.3\%$, physics and astronomy for $8.3\%$, and mathematics for $3.2\%$.

\label{sec:app-prior-record}
Direct linkage to DBLP shows that over two-fifths of the resolved 2024 entrants without earlier arXiv papers already have formal CS publications before their first arXiv submission (Table~\ref{tab:entrant-prior-record-dblp}). For all such entrants, the lower and upper bounds in Table~\ref{tab:entrant-prior-record-dblp} treat unresolved entrants as having no prior DBLP record or a prior DBLP record, respectively.

\input{tables/entrant_history_summary.tex}

\input{tables/entrant_prior_record_dblp.tex}

The retention comparison covers 158,870 first-time CS authors with an established CS coauthor and 15,331 on newcomer-only teams, and the first group has higher rates in all three outcomes (Table~\ref{tab:bowen-entry-return}).

\input{tables/entry_return_table.tex}

%% file: tables/entrant_history_summary.tex
\begin{table}
\centering
\caption{Growth of annual entrant counts by prior publication history, 2010--2024.}
\label{tab:entrant-history-summary}
\begin{tabular*}{\columnwidth}{@{\extracolsep{\fill}}lr@{}}
\toprule
\textbf{\boldmath History} & \textbf{\boldmath Growth} \\
\midrule
Prior CS & 6.9$\times$ \\
Prior non-CS & 19.9$\times$ \\
No prior record observed & 20.9$\times$ \\
Unresolved & 16.7$\times$ \\
\bottomrule
\end{tabular*}
\end{table}

%% file: tables/entrant_prior_record_dblp.tex
\begin{table}
\centering
\caption{Prior DBLP publications among 2024 CS entrants with no earlier arXiv record. All entries are percentages.}
\label{tab:entrant-prior-record-dblp}
\begin{tabular*}{\columnwidth}{@{\extracolsep{\fill}}lrrr@{}}
\toprule
\textbf{\boldmath Role} & \textbf{\boldmath Resolvable} & \textbf{\boldmath Prior / resolved} & \textbf{\boldmath Prior / all} \\
\midrule
All & 87.1 & 43.8 & 38.1--51.0 \\
First-author & 92.4 & 41.6 & 38.5--46.1 \\
\bottomrule
\end{tabular*}
\end{table}

%% file: tables/entry_return_table.tex
\begin{table}
\centering
\caption{Retention by entry-team composition. Rates average the 2010--2022 entry cohorts with equal weight.}
\label{tab:bowen-entry-return}
\begin{tabularx}{\columnwidth}{@{}>{\raggedright\arraybackslash}Xrrrr@{}}
\toprule
\textbf{\boldmath Entry team} & \textbf{\boldmath People} & \textbf{\boldmath Return} & \textbf{\boldmath Year 3} & \textbf{\boldmath \shortstack{Active\\years}} \\
\midrule
Established CS coauthor & 158,870 & 70.7\% & 41.4\% & 45.1\% \\
Newcomers only & 15,331 & 51.5\% & 27.7\% & 30.3\% \\
\bottomrule
\end{tabularx}
\end{table}

%% file: appendices/d_author_identity.tex
\section{Identity Resolution and Career Linkage}
\label{sec:app-identity-resolution}

The linkage procedure of Section~\ref{sec:identity} links nearly all author keys of arXiv-CS entrants to OpenAlex profiles (Table~\ref{tab:bowen-identity-coverage}). Shared identifiers and bibliographic evidence account for most accepted matches, and the trajectory model links additional authors with several papers. Entrants from 2025 have a higher unlinked share because external records index recent work with a delay.

On a held-out benchmark of 1,190 author keys and 5,397 groups of candidate OpenAlex profiles, the full linkage procedure, labeled Full framework in Figure~\ref{fig:bowen-identity-models}, ranks candidates more accurately than BEARD-fast \cite{Louppe2016} and Dedupe \cite{GreggEder2022}. Its area under the receiver operating characteristic (ROC) curve is the highest of the compared methods, and its acceptance decisions match the verified reference profiles with high precision and recall.

\input{tables/identity_coverage_table.tex}

\begin{figure}
\centering
\includegraphics[width=\columnwidth]{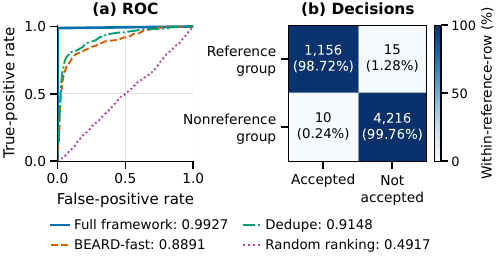}
\Description{ROC curves of the full linkage procedure, BEARD-fast, Dedupe, and random ranking, with the final acceptance decisions tabulated against reference matches.}
\caption{Held-out comparison of identity-linkage models and the final acceptance policy.}
\label{fig:bowen-identity-models}
\end{figure}

%% file: tables/identity_coverage_table.tex
\begin{table}
\centering
\caption{Identity-linkage routes and coverage. Route shares use all entrant keys, and multi-profile status overlaps acceptance routes.}
\label{tab:bowen-identity-coverage}
\begin{tabular*}{\columnwidth}{@{\extracolsep{\fill}}lrrr@{}}
\toprule
\textbf{\boldmath Route / coverage} & \textbf{\boldmath 2010--2024} & \textbf{\boldmath 2025} & \textbf{\boldmath All} \\
\midrule
Author keys & 647,354 & 131,322 & 778,676 \\
Direct-evidence route & 95.06\% & 79.50\% & 92.43\% \\
Trajectory route & 2.73\% & 9.39\% & 3.85\% \\
Unlinked & 2.22\% & 11.11\% & 3.72\% \\
\midrule
Linked total & 97.78\% & 88.89\% & 96.28\% \\
Multi-profile (overlap) & 16.63\% & 9.23\% & 15.39\% \\
\bottomrule
\end{tabular*}
\end{table}

%% file: appendices/c_output_distribution.tex
\section{Output Distribution Details}
\label{sec:app-distribution}

This appendix reports the distribution results behind Section~\ref{sec:productivity}.

Fractional output becomes more unequal while its population mean stays nearly unchanged. The population mean changes by $-1.0\%$ from 2021 to 2025, whereas fractional variance increases by $+60.8\%$ and the fractional Gini coefficient rises by $+11.3\%$, making 2021--2025 the largest four-year increase in both dispersion metrics since 2010.

The comparison of the same continuing authors covers 48,963 authors under fractional counting and 8,210 under first-author counting. Figure~\ref{fig:bowen-within-history} shows how their individual output moves in both directions within each historical output group.

\begin{figure*}
\centering
\includegraphics[width=\textwidth]{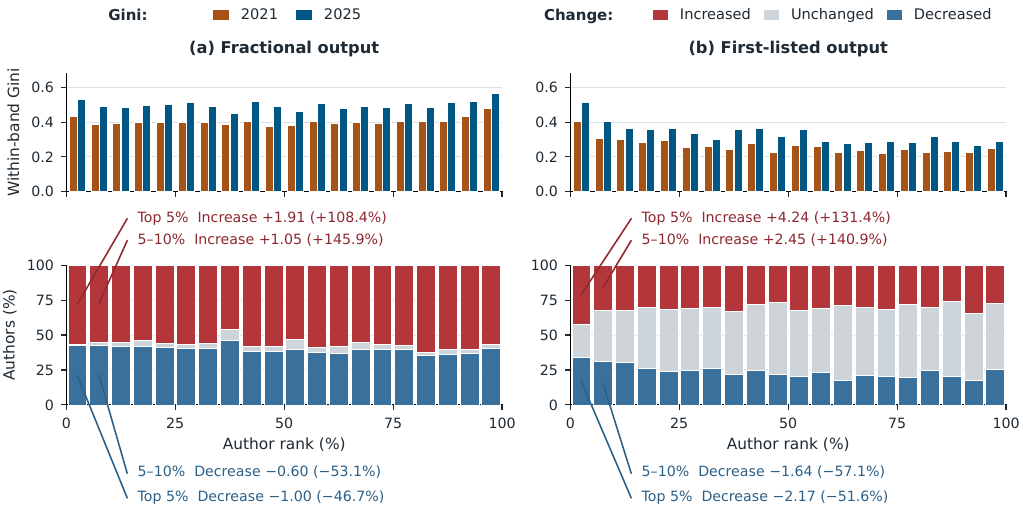}
\Description{Within-group Gini and shares of individual output increases, decreases, and unchanged output for continuing authors active in 2021 and 2025, assigned to fixed historical output groups.}
\caption{Output changes within historical output groups among the same continuing authors.}
\label{fig:bowen-within-history}
\end{figure*}

Within-group variance growth remains large when the weights of the twenty historical output groups are held at their 2021 shares of active members. Weighted within-group variance increases are $1.410$ under fractional counting and $3.258$ under first-author counting, closely matching the active-weighted decomposition in Table~\ref{tab:historical-variance-growth}.

\label{sec:app-mean-components}
Under fractional counting, the high-output tier adds 0.021 to the population mean from 2021 to 2025, while the other authors lower it by 0.026. Under first-author counting, both parts raise the mean (Table~\ref{tab:output-mean-components}).

\input{tables/output_mean_components.tex}

\label{sec:app-tenure}
The share of active authors in the low-output tier increases within every group of platform tenure between 2021 and 2025, rising by 7.4 to 9.9 percentage points (Table~\ref{tab:tenure-low-output}). Separating changes within tenure groups from changes in their population shares shows that increases within tenure groups add $8.71$ percentage points to the overall share, whereas shifts in tenure composition offset $1.17$ percentage points. Authors with a single CS paper carrying at least five authors make up $90.3\%$ of the low-output tier in 2025, which links the growth of this tier to larger coauthorship teams in every tenure group.

\input{tables/tenure_low_output.tex}

%% file: tables/output_mean_components.tex
\begin{table}
\centering
\caption{Contributions of output tiers to the change in population mean output, 2021--2025.}
\label{tab:output-mean-components}
\begin{tabular*}{\columnwidth}{@{\extracolsep{\fill}}lrrr@{}}
\toprule
\textbf{\boldmath Credit} & \textbf{\boldmath Above four} & \textbf{\boldmath Other} & \textbf{\boldmath Total} \\
\midrule
Fractional & $+0.0209$ & $-0.0260$ & $-0.0051$ \\
First-author & $+0.0608$ & $+0.0253$ & $+0.0861$ \\
\bottomrule
\end{tabular*}
\end{table}

%% file: tables/tenure_low_output.tex
\begin{table}
\centering
\caption{Shares of active authors in the low-output tier by platform tenure.}
\label{tab:tenure-low-output}
\begin{tabular*}{\columnwidth}{@{\extracolsep{\fill}}lrr@{}}
\toprule
\textbf{\boldmath Observed tenure} & \textbf{\boldmath 2021} & \textbf{\boldmath 2025} \\
\midrule
Entry year & 47.2\% & 57.1\% \\
Years 1--3 after entry & 30.3\% & 38.6\% \\
Year 4 onward & 19.9\% & 27.3\% \\
\bottomrule
\end{tabular*}
\end{table}

%% file: appendices/e_sources_sample_fields.tex
\section{Sample Summary and Category Definitions}
\label{sec:app-sources}
\label{sec:app-identity}
\label{sec:category-sets}

From 2010 to 2025, our arXiv sample covers more than 2.26 million distinct preprints and 1.63 million distinct authors in the four fields (Table~\ref{tab:analysis-scale}). Incidences ($H$) in Table~\ref{tab:analysis-scale} count author appearances on papers.

\input{tables/analysis_scale.tex}

The annual author-output records of Section~\ref{sec:identity} and the arXiv metadata counts under the inclusive field definitions differ by 3 CS papers in 2022 and by 44 in 2025. Analyses of activity paths, output distributions, and 2022--2025 growth use the author-output records, while long-run series use the metadata counts.

Broad AI comprises the eight categories \texttt{cs.AI}, \texttt{cs.CL}, \texttt{cs.CV}, \texttt{cs.IR}, \texttt{cs.LG}, \texttt{cs.MA}, \texttt{cs.NE}, and \texttt{cs.RO}. Mathematics includes any \texttt{math.*} label.

The Physics comparison field comprises the 13 arXiv archives \texttt{astro-ph}, \texttt{cond-mat}, \texttt{gr-qc}, \texttt{hep-ex}, \texttt{hep-lat}, \texttt{hep-ph}, \texttt{hep-th}, \texttt{math-ph}, \texttt{nlin}, \texttt{nucl-ex}, \texttt{nucl-th}, \texttt{physics}, and \texttt{quant-ph}.

Author counts and demographic shares change little under DBLP identity resolution. Name cleaning and resolution together reduce 2025 CS author counts by $1.4\%$ relative to raw names. Annual CS author counts differ by at most $0.30\%$ and entrant demographic shares by at most $0.20$ percentage points relative to normalized names.

%% file: tables/analysis_scale.tex
\begin{table}
\centering
\caption{Size of the retained arXiv analysis sample, 2010--2025. Field memberships overlap.}
\label{tab:analysis-scale}
\begin{tabular*}{\columnwidth}{@{\extracolsep{\fill}}lrrr@{}}
\toprule
\textbf{\boldmath Field} & \textbf{\boldmath \shortstack{Papers\\($P$)}} & \textbf{\boldmath \shortstack{Authors\\($A$)}} & \textbf{\boldmath \shortstack{Incidences\\($H$)}} \\
\midrule
Computer science & 852,481 & 787,999 & 3,560,065 \\
Broad AI & 542,584 & 569,869 & 2,505,497 \\
Mathematics & 620,747 & 297,685 & 1,423,674 \\
Physics & 1,055,753 & 904,413 & 7,446,670 \\
\midrule
Union of four fields & 2,260,466 & 1,635,624 & 11,638,959 \\
\bottomrule
\end{tabular*}
\end{table}

%% file: appendices/f_accounting_proofs.tex
\section{Proofs of the Accounting Identities and Growth Decomposition}
\label{sec:app-accounting-proofs}

\begin{proof}[Proof of Proposition~\ref{prop:paper-accounting}]
Reordering the finite sums gives
\begin{align*}
\sum_{a\in\mathcal{A}_t}q^A_{a,t}
&=\sum_{p\in\mathcal{P}_t}\sum_{a\in\mathcal{B}(p)}\frac{1}{n_p}=P_t,\\
\sum_{a\in\mathcal{F}_t}q^F_{a,t}
&=\sum_{p\in\mathcal{P}_t}1=P_t.
\end{align*}
Each paper contributes one unit of total credit under either rule. The definitions of the two means therefore give $A_tQ^A_t=F_tQ^F_t=P_t$. Finally,
\[
\frac{W_t}{\bar n_t}=\frac{H_t/A_t}{H_t/P_t}
=\frac{P_t}{A_t}=Q^A_t.
\]
All denominators are positive under the proposition's assumptions.
\end{proof}

\begin{proof}[Proof of Corollary~\ref{cor:growth-decomposition}]
Taking logarithms of the two paper-count identities gives $\log P_t=\log A_t+\log Q^A_t$ and $\log P_t=\log F_t+\log Q^F_t$. For any comparison period $\tau$, the geometric mean satisfies
\[
\log\bar X_\tau=\frac{1}{|\tau|}\sum_{t\in\tau}\log X_t.
\]
Averaging each identity over the same years with the same weights preserves its equality. Subtracting the early-period identity from the late-period identity gives Eqs.~\eqref{eq:all-author} and~\eqref{eq:first-author}. Single-year comparisons are the special case $|\tau|=1$. Division by $g_P\neq0$ gives the contribution shares. Annualizing all terms by the same interval preserves both decompositions and their shares.
\end{proof}